\documentclass{IEEEtran}

\usepackage{epsfig}
\usepackage{amsmath}
\usepackage{amsfonts}
\usepackage{amssymb}
\usepackage{subfigure}
\usepackage{booktabs}
\usepackage{citesort}
\usepackage{flushend}

\newtheorem{theorem}{Theorem}

\newtheorem{lemma}{Lemma}

\newcommand{\bsc}{{\rm bsc}}
\newcommand{\bec}{{\rm bec}}
\newcommand{\bsec}{{\rm bsec}}
\newcommand{\Er}{E_{\rm r}}
\newcommand{\Eex}{E_{\rm ex}}
\newcommand{\Esp}{E_{\rm sp}}

\newcommand{\Ex}{E_{ x}}

\newcommand{\Esc}{E_{\rm sc}}

\newcommand{\Yc}{{\cal Y}}
\newcommand{\EE}{\mathbb{E}}
\newcommand{\Cc}{{\cal C}}
\newcommand{\Xc}{{\cal X}}

\title{Extremes of Error Exponents}

\author{Albert Guill\'en i F\`abregas, Ingmar Land and Alfonso Martinez 
\thanks{A. Guill\'en i F\`abregas is with the Instituci\'o Catalana de Recerca i Estudis Avan\c{c}ats (ICREA), the Department of Information and Communication Technologies, Universitat Pompeu Fabra, Barcelona, Spain, and the Department of Engineering, University of Cambridge, Cambridge, United Kingdom, email: {\tt guillen@ieee.org}. I. Land is with the Institute for Telecommunications Research, University of South Australia, Adelaide SA 5001, Australia, e-mail: {\tt ingmar.land@ieee.org}. A. Martinez  is with the Department of Information and Communication Technologies, Universitat Pompeu Fabra, Barcelona, Spain, email: {\tt alfonso.martinez@ieee.org}.}
\thanks{The material in this paper was presented in part at the 2011 IEEE International Symposium on Information Theory, St. Petersburg, Russia, July-August 2011.}
\thanks{This work has been supported in part by the International Joint Project 2008/R2 of the Royal Society, by the Australian Research Council under ARC Discovery Grant DP0986089 and by the European Research Council under ERC grant agreement 259663. A. Martinez received funding from the Ministry of Science and Innovation (Spain) under grant RYC-2011-08150 and from the European Union's 7th Framework Programme (PEOPLE-2011-CIG) under grant agreement 303633.}
}

\begin{document}

\maketitle

\begin{abstract}
This paper determines the range of feasible values of standard error exponents for binary-input memoryless symmetric channels of fixed capacity $C$ and shows that extremes are attained by the binary symmetric and the binary erasure channel. The proof technique also provides analogous extremes for other quantities related to Gallager's $E_0$ function, such as the cutoff rate, the Bhattacharyya parameter, and the channel dispersion. 
\end{abstract}


\section{Introduction}

In the context of coded communication, the channel coding theorem relates the error probability and the code rate, showing that there exist codes whose error probability tends to zero provided that the code rate is smaller than the channel capacity. For uncoded systems, the error probability and the channel capacity are also related. In particular, references \cite{hellman1970probability,land2004reliability,land2006information} show that given one of the two values, tight bounds on the other can be given for the family of binary-input memoryless and symmetric (BIMS) channels. Such channels are described by the channel transition probability $P_{Y|X}(y|x)$, where $x\in\{x_0,x_1\}$ and $y\in\Yc$. We assume that the channel output alphabet $\Yc$ has finite size, though our approach also holds for well-behaved channels with infinite alphabet size, like the binary-input additive white Gaussian noise (BIAWGN) channel. We adopt Gallager's definition of symmetric channel \cite[p.\ 94]{gallager1968ita}, that is, a channel is said symmetric if the channel transition probability matrix (rows corresponding to input values) is such that it can be partitioned in submatrices for which each row is a permutation of any other row and each column is a permutation of any other column.
Both the binary erasure channel (BEC) and the binary symmetric channel (BSC) are symmetric.
 
More precisely, references \cite{hellman1970probability,land2004reliability,land2006information} show that the uncoded error probability of any BIMS channel with capacity $C$ is upper-bounded by that of the BEC and lower-bounded by that of the BSC of the same capacity. Similar results have been found in \cite{sason2009universal,arikan2009channel} for the Bhattacharyya parameter, a simple upper bound to the uncoded error probability; here only the extremal property of the BEC was proved. In the context of iterative decoding, analogous extremal properties of the BEC and BSC have been found \cite{land2005bounds,sutskover2005extremes} for the building blocks of iterative decoders for low-density parity-check codes, namely variable-node and check-node decoders.

Upper and lower bounds to the error probability  of good codes can be given in terms of error exponents, e.g.\ Gallager's random coding bound \cite[Thm.\ 5.6.3]{gallager1968ita}, the sphere-packing bound by Shannon {\em et al.} \cite{shannon1967lower} and Arimoto's strong converse bound \cite{arimoto1973converse}. These exponents are expressed as optimization problems involving Gallager's $E_0$ function \cite[Eq.\ 5.6.14]{gallager1968ita},
\begin{align}
E_0(\rho) \triangleq -\log F(\rho),
\end{align}
where 
\begin{equation}\label{eq:F0}
F(\rho) \triangleq \EE \left[ \left(\frac{\EE\bigl[P_{Y|X}(Y|X')^{\frac{1}{1+\rho}}|Y\bigr]}{P_{Y|X}(Y|X)^{\frac{1}{1+\rho}}}\right)^\rho\right]
\end{equation}
and the pair $(X,Y)$ is distributed according to $P_X P_{Y|X}$. Here and throughout the paper $\EE[\cdot]$ denotes the expectation of a random variable and all logarithms are in base $2$.

Equiprobable inputs maximize the $E_0$ function for BIMS channels \cite[p.\ 203]{jelinek1968probabilistic}, and we henceforth assume such distribution, i.e.\ $P_X(x_1) = P_X(x_2) = \frac{1}{2}$.

In this paper, we characterize the feasible values of $E_0(\rho)$ for an arbitrary BIMS channel of fixed channel capacity $C$ and show that the $E_0$ function is upper-bounded (resp.\ lower-bounded) by that of the BEC (resp.\ BSC) of the same capacity.  Since the aforementioned exponents are expressed using the $E_0$ function, we are able to find their extremal values. In fact, our analysis leads to similar results for the cutoff rate, the Bhattacharyya parameter, the channel dispersion and to a number of other extensions.

\section{Feasible Pairs of capacity $C$ and $F(\rho)$ Function}
\label{sec:E0_result}
 
The $F(\rho)$ functions for the BEC and BSC of  erasure/crossover probability $\varepsilon$, respectively denoted by $F^\bec(\rho)$ and $F^\bsc(\rho)$, are given by
\begin{align}
F^\bec(\rho) &\triangleq 2^{-\rho}(1-\varepsilon) + \varepsilon\\
F^\bsc(\rho) &\triangleq 2^{-\rho}\left(\varepsilon^{\frac{1}{1+\rho}}+ (1-\varepsilon)^{\frac{1}{1+\rho}}\right)^{1+\rho} .
\label{eq:F0_bsc_eps}
\end{align}
Using the capacity expressions for the BEC, $C^\bec \triangleq 1-\varepsilon$, and BSC, $C^\bsc \triangleq 1-h(\varepsilon)$,  we can find the erasure/crossover probability corresponding to a given capacity $C$ and parametrize the $F^\bec(\rho)$ and function $F^\bsc(\rho)$ as functions of $C$, namely
\begin{align}
&F^\bec(\rho;C) = 1+\left(2^{-\rho} -1\right)C, \label{eq:F_bec_C}\\
&F^\bsc(\rho;C) \notag\\
&~~~~=2^{-\rho}\left(\left(h^{-1}(1-C)\right)^{\frac{1}{1+\rho}}+ \left(1-h^{-1}(1-C)\right)^{\frac{1}{1+\rho}}\right)^{1+\rho}\label{eq:F_bsc_C}
\end{align}
where $h(p) \triangleq -p\log p - (1-p)\log(1-p)$ is the binary entropy function, and $h^{-1}(x)$ denotes the inverse of $h(p)$ for $p\in[0,\frac{1}{2}]$. $C^\bec\bigl(F(\rho)\bigr), C^\bsc\bigl(F(\rho)\bigr)$ are respectively defined as the inverses of Eqs. \eqref{eq:F_bec_C}, \eqref{eq:F_bsc_C} with respect to $C$.

For BIMS channels, one has the bounds $0 \leq C \leq 1$ and $0 \leq F(\rho) \leq 1$ for $\rho \geq 0$ and $1 \leq F(\rho) \leq 2$ for $-1 < \rho \leq 0$. This is a consequence of the facts that $F(\rho)$ is non-negative and non-increasing for $\rho > -1$ \cite[App. 5B]{gallager1968ita}, that $\lim_{\rho \to -1}F(\rho) = 2$, and that $F(0) = 1$.
It is however not apparent whether further limitations exist on the feasible pairs of capacity $C$ and $F(\rho)$. Against this first impression, the next theorem tightly characterizes the set of possible pairs of capacity $C$ and $F(\rho)$ function for any BIMS channel (see Figure \ref{fig:extremes_F0_c}). In the next section, we apply this theorem and prove several analogous characterizations for other relevant quantities in the analysis of the error probability over BIMS channels.
\begin{theorem}
\label{th:extremes_F0}
For any BIMS channel with capacity $C$ and function $F(\rho)$ for $\rho>-1$, the following statements hold
\begin{enumerate}
\item the function $F(\rho)$ of the channel satisfies
\begin{equation}
F^\bec(\rho;C) \leq F(\rho) \leq F^\bsc(\rho;C);
\label{eq:F0_extr_1}
\end{equation}
\item the capacity $C$ of the channel satisfies
\begin{align}
&C^\bsc\bigl(F(\rho)\bigr) \leq C \leq C^\bec\bigl(F(\rho)\bigr) \quad \text{for } -1 < \rho \leq 0,
\label{eq:F0_extr_2} \\
&C^\bec\bigl(F(\rho)\bigr) \leq C \leq C^\bsc\bigl(F(\rho)\bigr) \quad \text{for } \rho \geq 0.
\label{eq:F0_extr_3} 
\end{align} 
\end{enumerate}
The extremes in \eqref{eq:F0_extr_1}--\eqref{eq:F0_extr_3} are attained by the BEC and the BSC. 

Furthermore, for a given pair $\bigl(C,F(\rho)\bigr)$ satisfying the inequalities in \eqref{eq:F0_extr_1} or \eqref{eq:F0_extr_2}, there exists a BIMS channel with capacity $C$ and function $F(\rho)$. Conversely, if the inequalities do not hold for the pair $\bigl(C,F(\rho)\bigr)$, there exists no such BIMS channel with capacity $C$ and function $F(\rho)$.
\end{theorem}

\subsection{Proof of Theorem \ref{th:extremes_F0}}

The proof is built around the idea that every BIMS channel admits a decomposition into subchannels that are BSCs. This decomposition follows directly from Gallager's definition of symmetric channels \cite[p.\ 94]{gallager1968ita} as used in this paper. A formal description may be found e.g. in \cite{land2005bounds,land2006information}. Here we deem identical the BEC with erasure probability 1 and the BSC with crossover probability $\frac{1}{2}$. In this decomposition, each channel output $Y$ is associated with an index $A=f(Y)$ which is independent of the input and depends on the channel output only. We denote by $P_A(a)$ the probability mass or density function of subchannel $a$, and by $\Yc(a)$ the corresponding binary output alphabet of the BSC with index $a$. Assuming such a decomposition, and since $P_{Y|X}(y|x) = P_{Y|X,A}(y|x,a)P_A(a)$ \cite{land2005bounds,land2006information} we have
 \begin{align}
F(\rho) &= \EE \left[ \left(\frac{\EE\bigl[P_{Y|X}(Y|X')^{\frac{1}{1+\rho}}|Y\bigr]}{P_{Y|X}(Y|X)^{\frac{1}{1+\rho}}}\right)^\rho\right]\\
&=\EE\left[\EE \left[ \left(\frac{\EE\bigl[P_{Y|X,A}(Y|X',A)^{\frac{1}{1+\rho}}|Y,A\bigr]}{P_{Y|X,A}(Y|X,A)^{\frac{1}{1+\rho}}}\right)^\rho\Bigg| ~A\right]\right]\\
&=\EE \left[F^\bsc(\rho;C(A))\right],
\end{align}
where $C(a)$ denotes the capacity of subchannel $a$.

The following lemma is proved in Appendix \ref{app:F-is-concave}.

\begin{lemma}
  \label{lem:F-is-concave}
  The function $F^\bsc(\rho;C)$ is concave in $C\in [0,1]$ for any $\rho > -1$, non-decreasing for $-1 < \rho \leq 0$, and non-increasing for $\rho \geq 0$.
\end{lemma}

Noting that $\EE[C(A)] = C$, and given the concavity of the function $F^\bsc(\rho;C)$, we apply Jensen's inequality to obtain  
\begin{align}
    F(\rho) 
    &=  \EE \Bigl[ F^\bsc(\rho;C(A)) \Bigr]  \\
    &\leq F^\bsc(\rho; \EE[C(A)] )  \\
    &=   F^\bsc(\rho;C)  .
 \end{align}
The bound is obviously achieved when the channel is a BSC.  

Since $F^\bsc(\rho;C)$ is concave, we can lower-bound it by a straight line joining the points $F^\bsc(\rho;0)$ ($C=0$) and $F^\bsc(\rho;1)$ ($C=1$) (see Figure \ref{fig:extremes_F0_c}), and then evaluate the expectation, i.e., 
  \begin{align}
    F^\bsc(\rho;C)
    &\geq F^\bsc(\rho;0) + C\left( F^\bsc(\rho;1) - F^\bsc(\rho;0) \right)  \\
    &=   1 + C(2^{-\rho} - 1)    \\
    &=   F^\bec(\rho;C).
 \end{align}
This bound is obviously achieved when the channel is a BEC, thus proving Eq. \eqref{eq:F0_extr_1}.
\begin{figure}[t]
\begin{center}
\includegraphics[width=1.0\columnwidth]{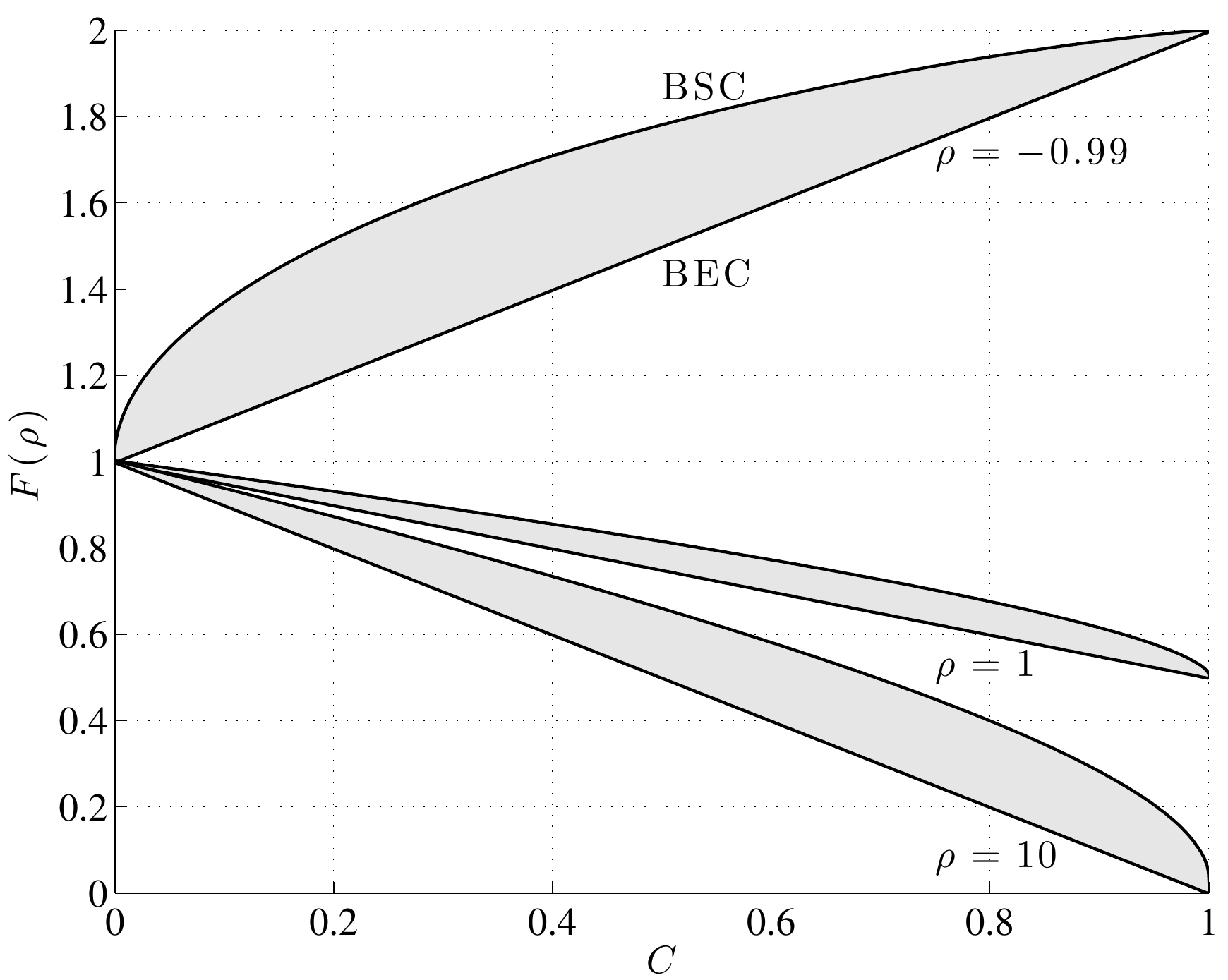}
\caption{Region of feasible points $\bigl(C,F(\rho)\bigr)$ for $\rho = -0.99, \, 1,\, 10$. The upper curves correspond to the BSC and the lower straight lines to the BEC.}
\label{fig:extremes_F0_c}
\end{center}
\end{figure}

Eq. \eqref{eq:F0_extr_1} determines the boundaries of the region of feasible pairs $\bigl(C,F(\rho)\bigr)$. Since $F^\bsc(\rho;C)$ is concave and $F^\bec(\rho;C)$ is convex, the region of feasible pairs is convex. Moreover, the functions $F^\bsc(\rho;C)$ and $F^\bec(\rho;C)$ are non-decreasing for $-1 < \rho \leq 0$ and non-increasing for $\rho \geq 0$. Fixing the value of $F(\rho)$, the convexity of the region implies Eq. \eqref{eq:F0_extr_2}.

We next prove that the region of feasible pairs $\bigl(C,F(\rho)\bigr)$ is connected by constructing a BIMS channel with corresponding capacity $C$ and function $F(\rho)$. Consider a binary symmetric-erasure channel (BSEC) with input alphabet $\{x_0,x_1\}$, output alphabet $\{y_0,y_1,y_\mathrm{e}\}$, cross-over probability $\varepsilon_\mathrm{s}$ and erasure probability $\varepsilon_\mathrm{e}$. Its transition probabilities are given by $P_{Y|X}(y_0|x_0) = P_{Y|X}(y_1|x_1) = 1-\varepsilon_\mathrm{e}-\varepsilon_\mathrm{s}$, $P_{Y|X}(y_0|x_1) = P_{Y|X}(y_1|x_0) = \varepsilon_\mathrm{s}$, and $P_{Y|X}(y_\mathrm{e}|x_0) = P_{Y|X}(y_\mathrm{e}|x_1) = \varepsilon_\mathrm{e}$. The capacity $C^\bsec$ and function $F^\bsec(\rho)$ are respectively
\begin{align}
  C^\bsec &= (1-\varepsilon_\mathrm{e})\Biggl(1 - h\biggl(\frac{\varepsilon_\mathrm{s}}{1-\varepsilon_\mathrm{s}-\varepsilon_\mathrm{e}}\biggr)\Biggr)\label{eq:C_bsec}\\
 F^\bsec(\rho) &= 2^{-\rho}\bigl(\varepsilon_\mathrm{s}^{\frac{1}{1 + \rho}} + (1 - \varepsilon_\mathrm{s} - \varepsilon_\mathrm{e})^{\frac{1}{1 + \rho}}\bigr)^{1 + \rho} +\varepsilon_\mathrm{e} .\label{eq:F_bsec}
\end{align}
For fixed $C$, there exist several BSEC channels with capacity $C$, among them a BSC and a BEC. Each of them is characterized by a pair of probabilities $\varepsilon_\mathrm{s}$ and $\varepsilon_\mathrm{e}$. The corresponding $F^\bsec(\rho)$ function is given by Eq. \eqref{eq:F_bsec}. Since the function $F^\bsec(\rho)$ is continuous in $\varepsilon_\mathrm{s}$ and $\varepsilon_\mathrm{e}$, one can always find a BSEC with capacity $C$ whose function $F^\bsec(\rho)$ coincides with the desired $F(\rho)$.

\subsection{Applications}
\label{sec:extensions}

In the proof of Theorem \ref{th:extremes_F0}, we exploited the fact that the region of feasible pairs $\bigl(C,F(\rho)\bigr)$ is convex and connected to characterize the extreme values of the capacity $C$ or the function $F(\rho)$. In this section, we apply the theorem to provide extreme values for other relevant quantities in the error probability analysis of channel coding. A simple extension to channel parameters $G$ given by $G(\rho) = g\bigl(F(\rho)\bigr)$, where $g(\cdot)$ is a monotonic continuous function, will prove convenient. 

\begin{theorem}
\label{th:extremes_general}
Let the channel parameter $G$ be given by $G(\rho) = g\bigl(F(\rho)\bigr)$, where $g(\cdot)$ is a monotonic strictly increasing continuous function. For any BIMS channel, we have that
\begin{enumerate}
\item the channel parameter $G(\rho)$ satisfies
\begin{equation}
g\bigl(F^\bec(\rho;C)\bigr) \leq G(\rho) \leq g\bigl(F^\bsc(\rho;C)\bigr); 
\label{eq:F0_extr_11}
\end{equation}
\item the channel capacity $C$ satisfies
\begin{align}
&C^\bsc\bigl(G(\rho)\bigr) \leq C \leq C^\bec\bigl(G(\rho)\bigr)\quad \text{for } -1 < \rho \leq 0,
\label{eq:F0_extr_22} \\
&C^\bec\bigl(G(\rho)\bigr) \leq C \leq C^\bsc\bigl(G(\rho)\bigr) \quad \text{for } \rho \geq 0.
\label{eq:F0_extr_33} 
\end{align} 
\end{enumerate}
The inequalities \eqref{eq:F0_extr_11}--\eqref{eq:F0_extr_33} are reversed if $g(\cdot)$ is monotonic, strictly decreasing and continuous. 
\end{theorem}

\paragraph*{Gallager's function} By letting $g(x) = -\log(x)$, the previous theorem readily gives the extremes of Gallager's function $E_0(\rho) = -\log F(\rho)$ for a fixed capacity, and the extremes of the capacity for a fixed $E_0$. 

\paragraph*{Cutoff rate} A particular case of the $E_0$ function is the cutoff rate, given by $R_0 = E_0(1)$. Thus, the above result also gives the extremes of the cutoff rate. 

\paragraph*{Bhattacharyya parameter}
A related quantity is the Bhattacharyya parameter $Z$, given by
\begin{align}
Z = \sum_{y\in\Yc} \sqrt{P_{Y|X}(y|x_0)P_{Y|X}(y|x_1)} = 2 F(1) -1.
\end{align}
The BSC/BEC have the largest/smallest possible Bhattacharyya parameter for BIMS channels of capacity $C$, interestingly giving the reverse extremes of the uncoded error probability \cite{hellman1970probability,land2004reliability,land2006information}. This result recovers Sason's \cite{sason2009universal} and Ar\i kan's \cite{arikan2009channel} bound for the BEC, and provides the extreme in the other direction attained by the BSC.
Fig.~\ref{fig:extremes_CZ} shows the bounds to $C$ for a given value of $Z$ from Theorem \ref{th:extremes_general}, as well as Ar\i kan's generic bounds for binary-input discrete memoryless channels \cite[Eqs. (1), (2)]{arikan2009channel}, illustrating some improvement.
\begin{figure}[t]
\begin{center}
\includegraphics[width=1.0\columnwidth]{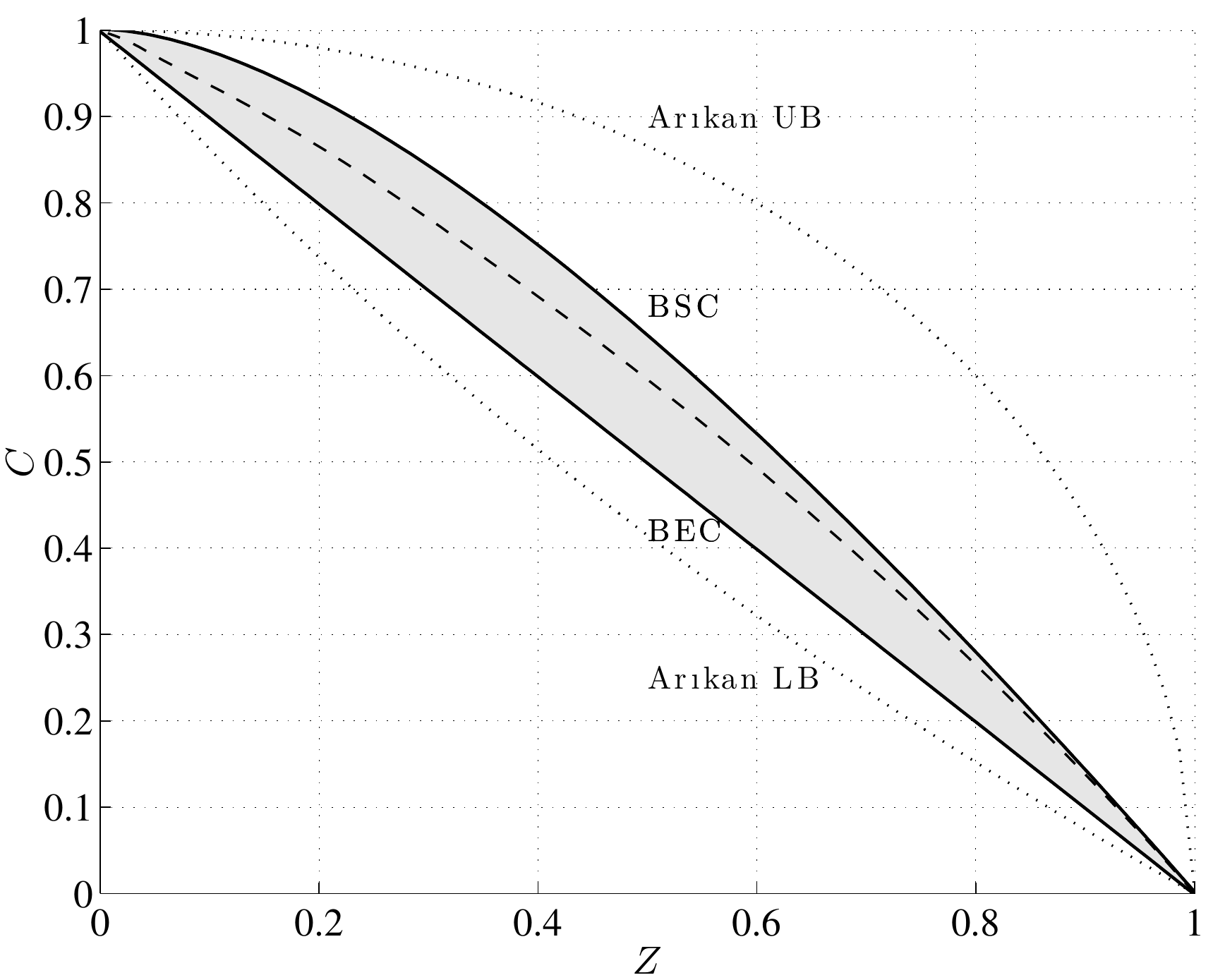}
\caption{Upper and lower bounds to the capacity $C$ as a function of the  Bhattacharyya parameter $Z$. Ar\i kan's upper and lower bounds \cite[Eqs. (1), (2)]{arikan2009channel} and the BIAWGN channel curve (dashed line) are also shown for reference. }
\label{fig:extremes_CZ}
\end{center}
\end{figure}

\paragraph*{Random coding exponent} The random coding exponent $\Er(R)$ \cite[Sect.\ 5.6]{gallager1968ita}, given by
\begin{equation}\label{eq:random_coding_exp}
\Er(R) = \max_{\substack{0\leq\rho \leq1}} E_0(\rho) -\rho R,
\end{equation}
provides an upper bound to the error probability of codes of rate $R$. This exponent involves a maximization of a function that, for fixed $\rho$ falls under the conditions for applicability of Theorems~\ref{th:extremes_F0} and~\ref{th:extremes_general}. Therefore, the exponent $\Er(R)$ satisfies
\begin{equation}
\Er^\bsc(R;C) \leq \Er(R) \leq \Er^\bec(R;C).
\label{eq:Er_extr}
\end{equation}
Figure \ref{fig:error_exponents_extremes_c} illustrates the extremes of random-coding error exponents $\Er(R)$. The random-coding error exponent of an arbitrary BIMS channel must lie in the shaded area, two such examples are the BIAWGN channel of the same capacity (with and without fading). 
\begin{figure}[htbp]
\begin{center}
\includegraphics[width=1.0\columnwidth]{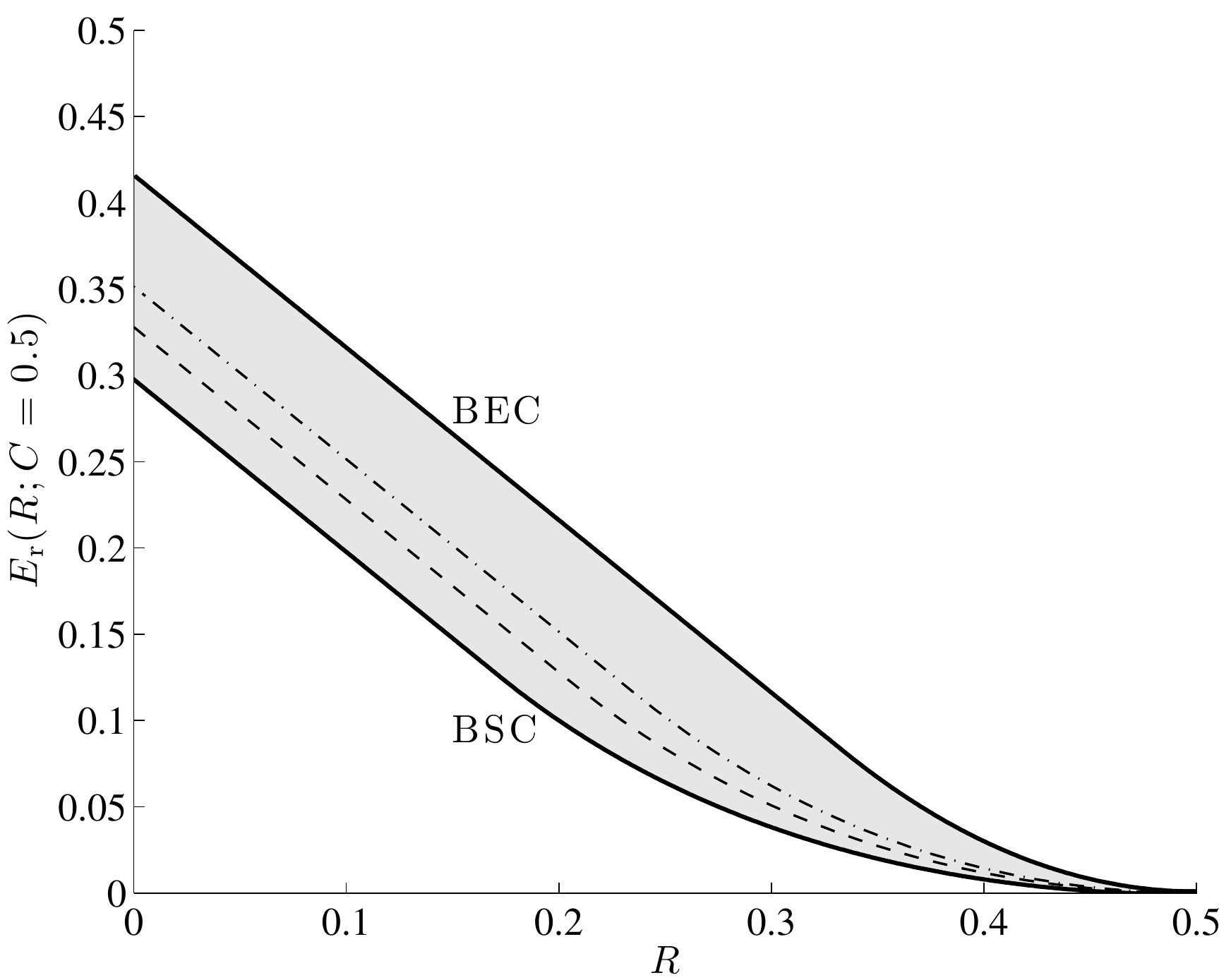}
\caption{Random coding error exponents of the BEC, BSC, BIAWGN (dashed), and Rayleigh fading BIAWGN (dash-dotted).}
\label{fig:error_exponents_extremes_c}
\end{center}
\end{figure}

\paragraph*{Expurgated error exponent} For rates below the channel critical rate, the expurgated error exponent $\Eex(R)$ \cite[Sect.\ 5.7]{gallager1968ita},  given by
\begin{equation}\label{eq:expurgated_coding_exp}
\Eex(R) = \max_{\substack{\rho \geq 1}} \Ex(\rho) -\rho R,
\end{equation}
provides a tighter estimate of the error probability of good codes than the random-coding exponent. The function $\Ex(\rho)$ is expressed in terms of the Bhattacharyya parameter $Z$ as
\begin{equation}\label{eq:expurgated_coding_exp0}
\Ex(\rho) = - \rho \log\frac{1+Z^{\frac{1}{\rho}}}{2}.
\end{equation}
Theorem~\ref{th:extremes_general} provides the extremes of the expurgated exponent.

\paragraph*{Strong converse exponent} In \cite{arimoto1973converse}, Arimoto lower-bounded the error probability of block codes at rates above capacity  in terms of the function $\Esc(R)$ given by
\begin{equation}
\Esc(R)\triangleq \sup_{-1<\rho \leq 0} E_0\left(\rho\right) - \rho R.
\label{eq:arimoto_exponent}
\end{equation}
Theorem~\ref{th:extremes_general} also provides the extremes of this exponent.

\paragraph*{Sphere-packing exponent} The error probability of codes of rate $R$ is lower-bounded by a bound that depends on the sphere-packing exponent \cite{shannon1967lower} $\Esp(R)$, given by
\begin{equation}
\Esp(R) \triangleq \sup_{\rho> 0} E_0(\rho)-\rho R.
\label{eq:sp_exponent}
\end{equation}
Again, Theorem~\ref{th:extremes_general} provides the extremes of this exponent.

\paragraph*{Threshold-decoding error exponents}
The exponent of random-coding bounds based on threshold decoding can also be expressed in closed form. Shannon \cite{shannon1957certain} derived the exponent of Feinstein's bound to the error probability \cite{feinstein1954new}. More generally, the exponent corresponding to a generalized form of Feinstein's bound \cite{martinez2011random} can be expressed as
\begin{equation}
\Er^{\rm gfb}(R) = \sup_{\rho\geq0} \frac{E_0(\rho) - \rho R}{1+\rho}.
\end{equation}
Theorem \ref{th:extremes_general} directly gives the error exponent extremes for the generalized Feinstein's bound.

The exponent of the dependence-testing bound \cite{polyanskiy2010channel} is \cite{martinez2011random}
\begin{align}
\Er^{\rm dtb}(R) &= \max_{0\leq \rho\leq 1} E_0(\rho,s=1) - \rho R,
\end{align}
where $E_0(\rho,s) \triangleq -\log F(\rho,s)$, for $s \geq 0$, and
\begin{equation}\label{eq:F0_1}
F(\rho,s) \triangleq \EE \left[ \left(\frac{\EE\bigl[P_{Y|X}(Y|X')^s|Y\bigr]}{P_{Y|X}(Y|X)^s}\right)^\rho\right].
\end{equation}
Following similar and somewhat simpler steps to those in the proof of Lemma \ref{lem:F-is-concave}, one can prove that $F(\rho,s=1)$, evaluated for a BSC with capacity $C$, is concave in $C$. Therefore, Theorem~\ref{th:extremes_general} holds and shows that the exponent of the DT bound has similar extreme values.

\paragraph*{Channel dispersion}
Recently, the Gaussian approximation to the error probability $P_e$ of length-$n$ codes at rates close to the capacity has received renewed attention. In this approximation, a critical channel parameter is the dispersion $V$, which for BIMS channels \cite{shannon1957certain,polyanskiy2010channel} is given by 
\begin{equation}
V = \frac{1}{\Er''(R=C)} = -E_0''(\rho=0).
\label{eq:V_E0}
\end{equation}
Moreover, it can be proved that one can choose either the $E_0(\rho)$ function or the simpler $E_0(\rho,1)$ to compute the latter derivative, that is $E_0''(\rho=0) = E_0''(\rho=0,1)$. As proved in Appendix \ref{app:bounded_i3}, the third derivative of $E_0(\rho,1)$ at $\rho=0$ is bounded for BIMS channels. Thus, a second-order Taylor expansion of $E_0(\rho,1)$ around $\rho =0$ shows that $E_0''(0,1)$ has the same extremes as $E_0(\rho,1)$. As illustration, Figure \ref{fig: extremes_dispersion_V} depicts the possible values of channel dispersion as a function of the capacity $C$ of the BIMS channel. The dashed line, which lies within the shaded area indicating the feasible region of pairs capacity/dispersion, corresponds to the BIAWGN channel.
\begin{figure}[t]
\begin{center}
\includegraphics[width=1.0\columnwidth]{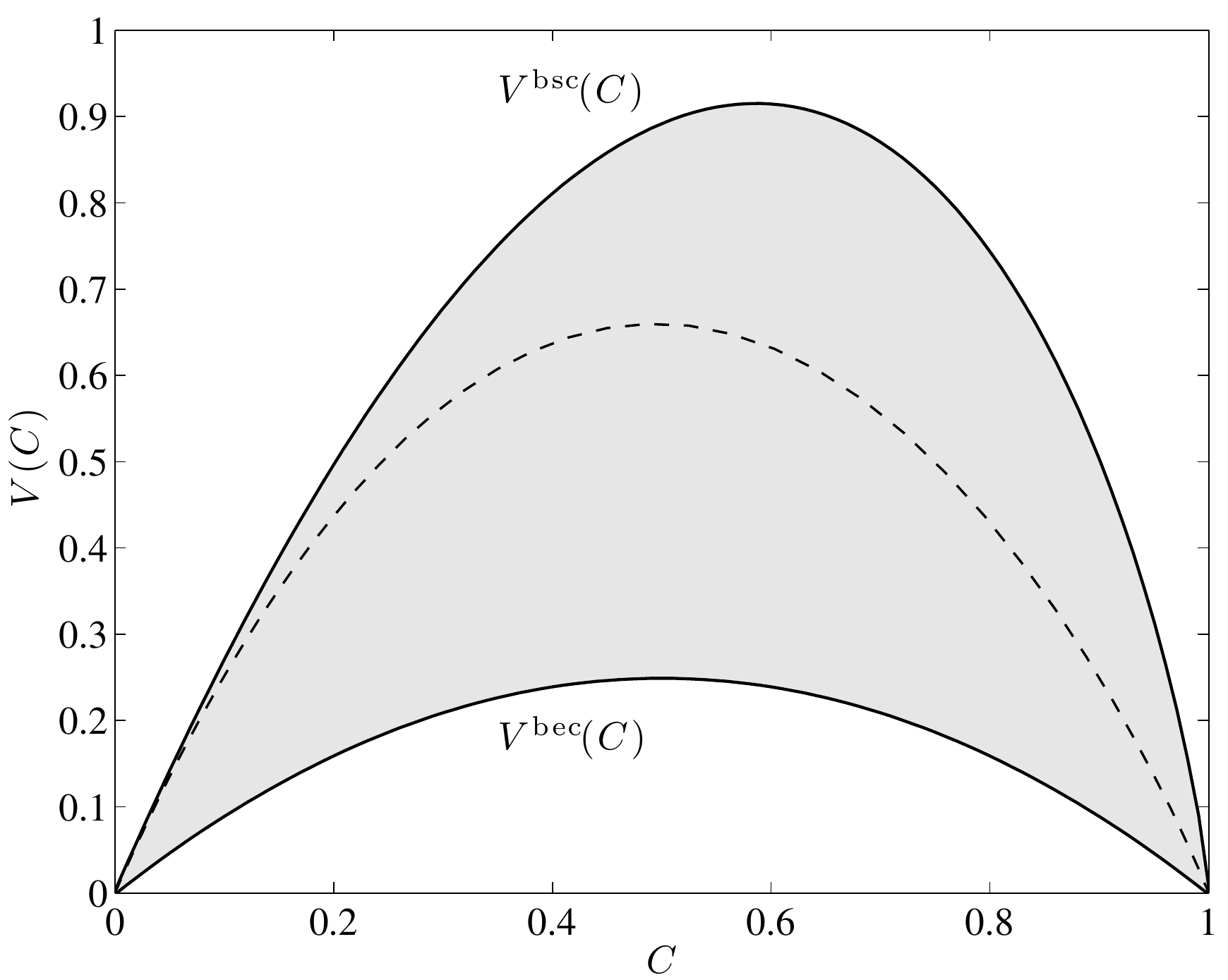}
\caption{Extremes of the channel dispersion $V(C)$.}
\label{fig: extremes_dispersion_V}
\end{center}
\end{figure}

\paragraph*{Error probability of specific codes} Our Theorems may also be applied to specific codes $\Cc$ with a given distance spectrum by means of the Shulman-Feder bound \cite{shulman1999random} (see also \cite{sason2006pal}) given by
\begin{equation}
-\frac{1}{n}\log P_e(\Cc) \geq \Er\left(R+\frac{1}{n}\log\alpha_\Cc\right)
\end{equation}
where $\alpha_\Cc$
is a function of the distance spectrum of the code that quantifies how far the distance spectrum of $\Cc$ is from that of the ensemble average.

\paragraph*{Exact error probability}
One might wonder whether our extremal results extend to the actual error probability. The answer is not immediately obvious. For uncoded transmission (a  code of length $n = 1$ and rate $R = 1$) over a given BIMS channel of capacity $C$ the error probability $P_b$ is upper- and lower-bounded by that of the BEC and the BSC, respectively, $h^{-1}(1-C) \leq P_b \leq \frac{1}{2}(1-C)$ \cite{land2004reliability}. In contrast, the extremes of the exponential bounds to the error probability, including the Bhattacharyya parameter, are reversed. This phenomenon suggests the existence of a pair $(n,R_n)$ such that a crossing point occurs, in the sense that for rates above (resp.\ below) $R_n$ the extremes may be those of uncoded transmission (resp.\ the error exponents).

\paragraph*{Connection with Ar\i kan, Telatar \cite{note_emre2008}, and Alsan \cite{alsan2012}} 
Unpublished work by Ar\i kan and Telatar \cite{note_emre2008} uncovered results of similar nature to those reported in this paper, showing that for channels with a fixed rate $R_\rho\triangleq \frac{\partial E_0(\rho)}{\partial \rho}$, for $0\leq \rho \leq 1$, the random coding exponent satisfies
\begin{equation}
\Er^\bec(R_\rho) \leq \Er(R_\rho) \leq \Er^\bsc(R_\rho). 
\end{equation} 
For $\rho = 0$ we have that $R_\rho = C$ and we obtain the trivial result that $0 = \Er^\bec(R_\rho) \leq \Er(R_\rho) \leq \Er^\bsc(R_\rho) = 0$. Instead, our results compare channels of a fixed capacity and provide the extremal values of the random-coding exponent and other quantities. The suitability of either of these two approaches to the problem may depend on the specific application. 
A more recent result by Alsan \cite{alsan2012} recovers both Theorem~\ref{th:extremes_F0} in this paper and the results in \cite{note_emre2008} as particular cases, for BIMS channels in the interval $0\leq \rho\leq 1$.

\appendices
\section{Proof of Lemma \ref{lem:F-is-concave}}
\label{app:F-is-concave}
We aim at proving the concavity of the function
\begin{equation}
f(C) = \left(\varepsilon(C)^{\frac{1}{1+\rho}}+ \bigl(1-\varepsilon(C)\bigr)^{\frac{1}{1+\rho}}\right)^{1+\rho}, 
\end{equation}
where $\varepsilon(C)$ is itself a function of $C$, namely $\varepsilon = h^{-1}(1-C)$. Without loss of generality, we limit our attention to the interval $\varepsilon\in[0,1/2]$.
The function is concave if $\frac{d^2f}{dC^2} \leq 0$.

Applying the chain rule of derivation, we have that
\begin{align}
 \frac{df}{dC} &= \frac{df}{d\varepsilon}\frac{d\varepsilon}{dC} \\
 \frac{d^2f}{dC^2} &= \frac{d^2f}{d\varepsilon^2}\biggl(\frac{d\varepsilon}{dC}\biggr)^2 + \frac{df}{d\varepsilon}\frac{d^2\varepsilon}{dC^2}.
\end{align}
Direct computation gives 
\begin{align}
 \frac{df}{d\varepsilon} &= \left(\varepsilon^{\frac{1}{1+\rho}}+ (1-\varepsilon)^{\frac{1}{1+\rho}}\right)^{\rho} \left(\varepsilon^{\frac{-\rho}{1+\rho}}- (1-\varepsilon)^{\frac{-\rho}{1+\rho}}\right)\\
\frac{d^2f}{d\varepsilon^2} &= -\frac{\rho}{1+\rho} \frac{\left(\varepsilon^{\frac{1}{1+\rho}}+ (1-\varepsilon)^{\frac{1}{1+\rho}}\right)^{\rho-1} }{\left(\varepsilon(1-\varepsilon)\right)^{\frac{1+2\rho}{1+\rho}}}.
\label{eq:2nd}
\end{align}
An application of the inverse function theorem yields 
\begin{align}
\frac{d\varepsilon}{dC} &= \frac{1}{\log\frac{\varepsilon}{1-\varepsilon}}\\
\frac{d^2\varepsilon}{dC^2} &= - \frac{1}{\varepsilon(1-\varepsilon)\left(\log\frac{\varepsilon}{1-\varepsilon}\right)^3 \ln 2}.
\end{align}
The derivatives with respect to $C$ are therefore given by
\begin{align}
&\frac{df}{dC} = \left(\varepsilon^{\frac{1}{1+\rho}}+ (1-\varepsilon)^{\frac{1}{1+\rho}}\right)^{\rho} \left(\varepsilon^{\frac{-\rho}{1+\rho}}- (1-\varepsilon)^{\frac{-\rho}{1+\rho}}\right) \frac{1}{\log\frac{\varepsilon}{1-\varepsilon}}\\
&\frac{d^2f}{dC^2} = -\frac{\left(\varepsilon^{\frac{1}{1+\rho}}+ (1-\varepsilon)^{\frac{1}{1+\rho}}\right)^{\rho-1} }{ \left(\log\frac{\varepsilon}{1-\varepsilon}\right)^2~ \left(\varepsilon(1-\varepsilon)\right)^{\frac{1+2\rho}{1+\rho}}} \notag\\
&\left(\frac{\rho}{1+\rho} + \frac{1-2\varepsilon + \varepsilon^{\frac{1}{1+\rho}} (1-\varepsilon)^{\frac{\rho}{1+\rho}} -\varepsilon^{\frac{\rho}{1+\rho}} (1-\varepsilon)^{\frac{1}{1+\rho}}}{ \ln\left(\frac{\varepsilon}{1-\varepsilon}\right)}\right).
\end{align}
Since we have that $\frac{df}{dC} \geq 0$ for $-1 < \rho \leq 0$ and $\frac{df}{dC} \leq 0$ for $\rho \geq 0$, we conclude that $f(C)$ is increasing and decreasing in the respective ranges of $\rho$.

The term before the brackets
\begin{equation}
-\frac{\left(\varepsilon^{\frac{1}{1+\rho}}+ (1-\varepsilon)^{\frac{1}{1+\rho}}\right)^{\rho-1} }{ \left(\log\frac{\varepsilon}{1-\varepsilon}\right)^2~ \left(\varepsilon(1-\varepsilon)\right)^{\frac{1+2\rho}{1+\rho}}} 
\end{equation}
is always non-positive for $-1\leq \rho\leq \infty$. Therefore, it suffices to show that the function
\begin{equation}
g(\varepsilon,\rho)  \triangleq \frac{\rho}{1+\rho} + \frac{1-2\varepsilon + \varepsilon^{\frac{1}{1+\rho}} (1-\varepsilon)^{\frac{\rho}{1+\rho}} -\varepsilon^{\frac{\rho}{1+\rho}} (1-\varepsilon)^{\frac{1}{1+\rho}}}{ \ln\left(\frac{\varepsilon}{1-\varepsilon}\right)}
\label{eq:g_func}
\end{equation}
is non-negative for $\varepsilon \in [0,\frac{1}{2}]$ and $-1\leq \rho\leq \infty$. 

Let $z\triangleq \frac{\varepsilon}{1-\varepsilon} \in [0,1]$. With this change of variables we obtain
\begin{equation}
 g(z,\rho) = \frac{\rho}{1+\rho}+ \frac{1-z+z^{\frac{1}{1+\rho}} - z^{\frac{\rho}{1+\rho}}}{(1+z)\ln z}.
\label{eq:g_z}
\end{equation}
We wish to show that $g(z,\rho)\geq 0$. The partial derivative with respect to $\rho$ is given by
\begin{align}
\frac{\partial g(z,\rho)}{\partial \rho} &= \frac{1}{(1+\rho)^2}\left(\frac{1+z-  z^{\frac{1}{1+\rho}} -  z^{\frac{\rho}{1+\rho}}}{1+z}\right)\\
&\triangleq \frac{1}{(1+z)(1+\rho)^2}~g_0(z,\rho)
\end{align}
We are interested in the sign of $g_0(z,\rho)$, whose derivative is in turn given by
\begin{align}
\frac{\partial g_0(z,\rho)}{\partial \rho} &= \frac{ \ln z}{(1+\rho)^2}\left(z^{\frac{1}{1+\rho}} - z^{\frac{\rho}{1+\rho}}\right).
\end{align}
We readily see that
\begin{align}
\frac{\partial g_0(z,\rho)}{\partial \rho} &\geq 0, ~~~~~\rho\in(-1,1],\\ 
\frac{\partial g_0(z,\rho)}{\partial \rho} &\leq 0,~~~~~\rho\in[1,+\infty).
\end{align}

Summarizing, since $g(z,\rho)$ is continuous in $\rho$ for $\rho>-1$, we have that
\begin{itemize}
\item $\frac{\partial g(z,\rho)}{\partial \rho} \leq 0$ in $\rho\in(-1,0]$,  since $g_0(z,\rho)$ is non-decreasing and $g_0(z,\rho) \leq g_0(z,0)=0$,
\item $\frac{\partial g(z,\rho)}{\partial \rho} \geq 0$ in $\rho\in[0,1]$, since $g_0(z,\rho)$ is non-decreasing and $g_0(z,\rho) \geq g_0(z,0)=0$,
\item $\frac{\partial g(z,\rho)}{\partial \rho} \geq 0$ in $\rho\in[1,\infty)$, since $g_0(z,\rho)$ is non-increasing and $g_0(z,\rho)\geq \displaystyle\lim_{\rho\to\infty }g_0(z,\rho)=0$.
\end{itemize}
The fact that $g(z,0)=0$ concludes the proof.

\section{}
\label{app:bounded_i3}

We wish to prove that the partial derivative $\frac{\partial^3E_0(\rho,s=1) }{\partial \rho^3}\bigr |_{\rho=0}$ is bounded. To this end, we first note that the function $E_0(\rho,s=1)$ can be expressed as
\begin{align}
E_0(\rho,s=1) 
&= -\log \EE\left[2^{-\rho \,i(X;Y)}\right],
\end{align}
where $i(x;y)$ is the information density, defined as
\begin{equation}
i(x;y) \triangleq \log \frac{P_{Y|X}(y|x)}{P_Y(y)}.
\end{equation}
The function $E_0(\rho,s=1)$  is a cumulant generating function. Its third derivative evaluated at $\rho = 0$ gives the third-order cumulant, that is the third-order central moment,
\begin{equation}
\frac{\partial^3E_0(\rho,s=1) }{\partial \rho^3}\biggr |_{\rho=0} = \EE\left[\bigl(i(X;Y) - I(X;Y)\bigl)^3\right] (\ln 2)^2.
\end{equation}

The next result shows that the $k$-th absolute moment of the information density is bounded.
\begin{lemma}
\label{lemma:mink}
Consider a memoryless channel with discrete input alphabet $\Xc$ and arbitrary output alphabet $\Yc$. Then, with equiprobable inputs we have
\begin{equation}
\EE\left[\bigl|i(X;Y) - I(X;Y)\bigl|^k\right]\leq \left(2\log|\Xc| + \frac{k}{\ln2}\bigl(1 + |\Xc|^{\frac{1}{k}}\bigr)\right)^k.
\end{equation}
\end{lemma}

\begin{proof}
We will make use of Minkowski's inequality
$
\|A+B\|_k\leq  \|A\|_k + \|B\|_k
$
where
$
\|A\|_k \triangleq (\EE[|A|^k])^{\frac{1}{k}}.
$
Using the definition of $i(X,Y)$, we now have that
\begin{align}
&\|i(X;Y)- I(X;Y)\|_k \notag\\
&\leq\left\| \log \frac{\sum_{x'}\frac{1}{|\Xc|}P_{Y|X}(Y|x')}{P_{Y|X}(Y|X)}\right\|_k+ I(X;Y)\\
&\leq 2\log |\Xc| + \left\| \log \frac{\sum_{x'}P_{Y|X}(Y|x')}{P_{Y|X}(Y|X)}\right\|_k\\
&\leq 2\log |\Xc|  + \frac{1}{\ln2}\left\| k \left(\frac{\sum_{x'}P_{Y|X}(Y|x')}{P_{Y|X}(Y|X)}\right)^\frac{1}{k} -k\right\|_k    \\
&\leq  2\log |\Xc|  + \frac{k}{\ln2} + \frac{k}{\ln2} \left( \EE\left[\frac{\sum_{x'}P_{Y|X}(Y|x')}{P_{Y|X}(Y|X)}\right]\right)^\frac{1}{k}\\
&\leq  2\log |\Xc|  + \frac{k}{\ln2} + \frac{k}{\ln2} |\Xc|^\frac{1}{k}
\end{align}
where we have used that $\ln x \leq k(x^\frac{1}{k}-1)$ \cite[Eq. (4.1.37)]{abramowitz1964hmf}. 
\end{proof}

Using Lemma \ref{lemma:mink}, we have that for BIMS channels, 
\begin{equation}
\frac{\partial^3E_0(\rho,s=1) }{\partial \rho^3}\biggr |_{\rho=0} \leq (\ln 2)^2\left(2 + \frac{3}{\ln2}(1+2^{\frac{1}{3}})\right)^3 .
\end{equation}

\end{document}